\documentclass[12pt]{article}
\pdfoutput=1
\usepackage{graphicx}
\usepackage{bbm}
\usepackage{amsmath}
\usepackage{theorem} 
\usepackage{amsfonts}
\usepackage{amssymb}
\usepackage{accents}
\usepackage[onehalfspacing]{setspace}
\usepackage{hyperref}
\usepackage{authblk}

\usepackage{fullpage}

\newtheorem{theorem}{Theorem}

\newtheorem{assumption}[theorem]{Assumption}

\newtheorem{condition}[theorem]{Condition}

\newtheorem{corollary}[theorem]{Corollary}

\newtheorem{definition}[theorem]{Definition}

\newtheorem{proposition}[theorem]{Proposition}

\newenvironment{proof}[1][Proof]{\noindent\textbf{#1.} }{\ \rule{0.5em}{0.5em}}
\begin{document}
	
\title{Fragility of Confounded Learning
}

\author[]{Xuanye Wang \thanks{Email: xuanye.wang@dufe.edu.cn.}}

\affil[]{\footnotesize Institute for Advanced Economics Research, Dongbei University of Finance and Economics}

\date{\today}

\maketitle
\begin{abstract}
	We consider an observational learning model with exogenous public payoff shock. We show that confounded learning doesn't arise for almost all private signals and almost all shocks, even if players have sufficiently divergent preferences. 
\end{abstract}


	
\section{Introduction}
People herd in daily lives. Everyone goes to their friends' favorite restaurants. College students are more likely to register the same classes as their familiar senior students do.  Intuitively, one may think such herding behavior as blind and irrational.

In contrast to intuition, \cite{Banerjee1992} and \cite{BHW1992} point out that herding behavior can be rational. In many circumstances, people's payoff are affected by an unknown underlying state. For example, the payoff of choosing one restaurant depends on the underlying quality of the chef. Decisions made by previous choice-makers reveal valuable private information about the underlying state. Thus, one could rationally herd as long as one thinks the observed previous choices contain enough information.

These two seminal works start a large literature named observational learning. This literature focuses on circumstances where learning could only happen through observing others' actions. In these circumstances, information about the underlying state is distributed among a large group of people. Each person's private information cannot be credibly or precisely transmitted. People guess others' private information behind observed actions, and then combine his/her own private information with these guesses to make decisions. One learns from others' actions, and one's own action reveals one's private information that others can learn. The key question is whether long run efficiency is achieved, whether people eventually assign all the weight to the true underlying state among all alternatives.

The answer depends on several important primitives of modeling: (1) the observing structure-whose actions one person could observe; (2) the signal structure-what kind of private information one person could get; (3) the payoff structure-what payoff one person could receive. Lots of important papers study different variations of these primitives and provide excellent answers.  To name a few, \cite{SS2000} study the classic observing structure that each player observes all previous actions in order. They find that long run effciency critically depend on whether players can get arbitrarily precise signal about the underlying states. If they do, then any wrong herd is eventually overthrown by
some player with sufficiently precise information. Correct learning must happen. \cite{CK2004} study the observing structure  that each player just observe the action of his/her immediate predecessor. They find that people's choices and beliefs cycle. In other words, learning needs not to happen. \cite{MAT2015} place strategic players on a directed social network and let them take actions repeatedly. They assume a player $A$ could only observe the action of player $B$ iff a directed edge connected them. They find that long run efficiency occurs as long as the number of neighborhoods of every player is bounded above and each player can observe his/her observers by a chain with bounded length.

We think one simple but important variation is missed in the literature-the payoff structure may be subject to exogenous shocks in the learning process. For example, the taste of the main course could be different in different seasons for the reason that some key ingredients tastes worse in winter. The benefits of taking the same course may change from year to year as the instructor changes. 

Based on this consideration, we introduce exogenous public payoff shocks into observational learning models. There are certainly many different ways to model payoff shocks. In this article we consider the most simple one. We assume there are two infinite sequence of periods. In one sequence, the payoff of one action is changed by a payoff shock with value $c$. Even such a simple structure of payoff shocks makes a difference: 
it makes confounded learning fragile even if players have sufficiently divergent preferences.

To illustrate the idea, let us briefly review the concepts of learning and confounded learning. The first thing we would like to mention is that the usage of the word ``long run learning" in observational learning literature is different from the usage in daily life and possibly other literatures. A ``long run learning belief" is often just a stationary belief over all potential states and needs not to assign all the weight to the true state. Reaching a long run learning belief doesn't mean efficiency-people need not to assign all the weight to the true state. The well-known ``wrong herd" in \cite{Banerjee1992} is associated with a long run learning belief which assigns sufficient weight to the wrong state. One of the many contributions of 
\cite{SS2000} is that inefficiency at a long run learning belief is not necessarily associated with wrong herd-the learning could be confounded. The key feature of confounded learning is its inconclusiveness-the previous actions don't provide overwhelming evidences about which state is true, and each person still has to use their own private information to decide. A sharp reader may question such a belief cannot be stationary: since each person still use private information to decide, his/her action contains certain amount of information; upon observing such an informative action, the belief must be updated and cannot stay the same! The smart point that \cite{SS2000} point out is that an observable action needs not to be informative even if the action-taker uses private information to decide. Actions are uninformative as long as the probability of the observed action is the same across different states. Mathematically, it means
\begin{eqnarray}
	\label{eqn1}
	\Pr(b|\lambda^*,A)=\Pr(b|\lambda^*,B),
\end{eqnarray}
the probability of observing action $b$ under confounded belief $\lambda^*$ is the same across true state $A$ and alternative state $B$. Of course equation \ref{eqn1} is just necessary but not sufficient for a belief to be confounded. A confounded belief must also be moderate in the sense that this belief doesn't dominate one's private information. \cite{SS2000} shows that confounded learning robustly arise in the sense that equation \ref{eqn1} robustly has moderate solutions.

The mathematics behind confounded learning changes after we introduce the simple public payoff shock. Now for a belief $\lambda^*$ to be confounded, any observable action must be uninformative under belief $\lambda^*$ with and without shock $c$. Mathematically, that means the following equation system
\begin{eqnarray}
	\label{eqn2}
	\Pr(a|\lambda^*,A,0)=\Pr(a|\lambda^*,B,0),\notag\\
	\Pr(a|\lambda^*,A,c)=\Pr(a|\lambda^*,B,c).
\end{eqnarray}
has moderate solutions. 
Despite that each equation in the system robustly has moderate solutions, we prove that the two equations almost never share a moderate solution. In other words, confounded learning almost never arise in models with payoff shock.

The existence of confounded learning undermines long run efficiency. In fact, in a model without payoff shock, its confounded learning beliefs may be uncountable and even homeomorphic to a cantor set. (See \cite{Wang2021}). However, no such belief could stay stationary as long as an arbitrary small shock $c$ arrives regularly. In this sense, we say confounded learning is fragile under payoff shock. 

This article primarily contributes to the observational learning literature by proving that confounded learning is fragile to almost all small payoff shock in a widely considered complete information, divergent preferences setting. In reaching that conclusion, we develop a math framework which may be useful in proving other generic results in observational learning models. Some of our math results, like private signals with real-analytic density functions are dense in certain sense, may have independent usage for other researchers.

Our work is sharply different from a recently a growing literature on observational learning with payoff externalities, which focus on long run efficiency in an environment that players' payoffs depend on not only the underlying state but also actions of other players. See \cite{AK2012}, \cite{EGKR2014}, \cite{Arieli2017}, \cite{Mozon2019} and \cite{SZ2020}. 
Our work is different both from the set up perspective and the result perspective.  The payoff changes in this literature is endogenous. In contrast, the payoff shock in our model is exogenous. What is more important is that this literature does not discuss confounded learning. They aims at study whether long run correct learning arises under different conditions. In contrast, our result is a generic result that confounded learning is fragile under exogenous public payoff shock.

This article is organized as following: in section 2 we set up the model; in section 3 we solve it; in section 4 we describe different types of long run learning in our model; in section 5 we rigorously state the results; in section 6 we prove.

\section{Model with Public Payoff Shock}
\label{section2}
Our model is based on the classic model in \cite{SS2000} where confounded learning arises due to players having divergent preferences. We twist the model by assuming that in each period a payoff shock of value $c$ arrives with a commonly known positive probability. The shock changes the payoff of action $b$ by value $c$. Below is a detailed description.

The model is of discrete-time. In period $0$, nature chooses one realization out of two potential states $A$ and $B$ according to a flat prior. In each period $t\geq 1$, one player arrives and needs to make a choice between actions $a$ and $b$. There are two types of players-``match" type and ``mismatch" type. The match type receives positive payoff iff his/her actions matches the realized underlying state. The mismatch type receives positive payoff if the chosen action is different from nature's chosen state. In each period, a payoff shock arrives with a commonly known probability. The arrival of payoff shock is independent from other primitives of the model, and the specific arrival probability doesn't matter for our results.
Table \ref{noshock} gives the payoff matrices of match type (M) and mismatch type (MM) without payoff shock. Table \ref{shock} gives the payoff matrices with shock. Here $u,v>0$ and $u\neq v$. The payoff tables are commonly known by all players.
\begin{table}[h]
	\centering
	\begin{tabular}{c c c c c c c c c}
		M & $a$ & $b$ & & & &       MM   & $a$ & $b$\\
		$A$ & $u$   & $0$   & & & & $A$ & $0$   & $v$\\
		$B$ & $0$   & $1$   & & & & $B$ & $1$   & $0$\\
	\end{tabular}
	\caption{Payoff without payoff shock}
	\label{noshock}
\end{table}
\begin{table}[h]
	\centering
	\begin{tabular}{c c c c c c c c c}
		M & $a$ & $b$ & & & &       MM   & $a$ & $b$\\
		$A$ & $u$   & $c$   & & & & $A$ & $0$   & $v+c$\\
		$B$ & $0$   & $1+c$   & & & & $B$ & $1$   & $c$\\
	\end{tabular}
	\caption{Payoff without payoff shock}
	\label{shock}
\end{table}

Before deciding, player at period $t$ observes the ordered choices made by his predecessors and whether each predecessor meets a payoff shock. He/she doesn't observe the type of each predecessor. But players commonly know that each player is of match type with probability $p$. Player $t$ observes his/her own type and payoff table before deciding.

Besides, player $t$ also observes a noisy private signal about nature's choice. This noisy signal is a state-contingent random variable $\mathcal{S}$. Its distribution is $F^A(s)$ if nature chooses state $A$, and is $F^B(s)$ if nature chooses state $B$.
Following the literature,  private signal is modeled as a direct signal, that is, $\mathcal{S}=s$ means that `` the probability that nature chooses state $A$ is $s$". We also assume that there is no fully revealing signal. For technical reason, we make the following assumption
\begin{assumption}[Continuous and Uniformly Bounded Assumption]
	\label{assumption1}
	We assume the distributions of private signal $\mathcal{S}$ under both states $(F^A(s),F^B(s))$ have continuous and uniformly bounded density functions $(f^A(s),f^B(s))$. 
\end{assumption}
Besides, all players commonly know the state-contingent distributions of signal $S$ and that $\{S_t\}$ is i.i.d. across periods. 

To summarize, the information set of player $t$ before deciding is 
\begin{eqnarray}
	\{S_t,T_t,h_t\}.
\end{eqnarray}
Here $S_t$ is his/her private signal's realization, $T_t$ is his/her type realization. These are private information known only to player $t$. That $h_t$ is the public history up to period $t$. It is an ordered sequence
\begin{eqnarray}
	h_t=\{(\alpha_{1},\chi_{1}),\dots, (\alpha_{t-1},\chi_{t-1}), \chi_t\}.
\end{eqnarray}
Here $\alpha_k\in \{a,b\}(1\leq k\leq t-1)$ is the action chosen at period $t$, $\chi_k$ denotes whether a payoff shock $c$ occurs at period $k$. Public history $h_t$ is observed not only by player $t$, but also all players arrives later.

We shall refer the above model as a \textit{shock model}. We often need to compare it with the corresponding \textit{no shock model}.
By \textit{corresponding}, we mean the no shock model has the same primitives as the shock model except that payoff is described by table \ref{noshock} in every period. Roughly speaking, our result says all the confounded learning beliefs of the no shock model turns non-stationary in the corresponding shock model for almost all small shocks.

\section{Solution of the Shock Model}
The shock model can be solved through a standard martingale argument. The main conclusion is that long run learning still happens almost surely, as in any no shock model. Readers who are not familiar with the standard martingale argument in solving an observational learning model can read \cite{SS2000} and \cite{Wang2019}. 

Without loss of generality, let us assume that nature chooses state $A$ in period $0$. The public belief $\lambda_t$ at period $t$ is the posterior likelihood ratio of state being $B$ over being $A$ conditional on public history $h_t$:
\begin{eqnarray}
	\lambda_t=\frac{\Pr(B|h_t)}{\Pr(A|h_t)}.
\end{eqnarray}
Player $t$'s posterior belief of state being $B$ over being $A$ is
\begin{eqnarray}
	\frac{\Pr(B|h_t,S_t)}{\Pr(A|h_t,S_t)}=\lambda_t\frac{1-S_t}{S_t}.
\end{eqnarray}
So we can compute the decision rule of player $t$. Player $t$ follows a cutoff strategy with the cutoff described as following:
\begin{itemize}
		\item If shock arrives at period $t$, a match type chooses action $b$ iff
		\begin{eqnarray}
			S_t<\frac{\lambda_t(1+c)}{\lambda_t(1+c)+(u-c)}\equiv m(\lambda_t,c).
		\end{eqnarray} 
		If no shock arrives at period $t$, a match type chooses action $b$ iff
		\begin{eqnarray}
			S_t<\frac{\lambda_t}{\lambda_t+u}= m(\lambda_t,0).
		\end{eqnarray} 
		\item 
		If shock arrives at period $t$, a mismatch type chooses action $b$ iff
		\begin{eqnarray}
			S_t>\frac{\lambda_t(1-c)}{\lambda_t(1-c)+(v+c)}\equiv	mm(\lambda_t,c).
		\end{eqnarray} 
		If no shock arrives at period $t$, a mismatch type chooses action $b$ iff
		\begin{eqnarray}
			S_t>\frac{\lambda_t}{\lambda_t+v}\equiv mm(\lambda_t,0).
		\end{eqnarray} 
\end{itemize}
After player $t$ makes his/her decision, the public belief evolves as following
\begin{eqnarray}
	\lambda_{t+1}=\frac{\Pr(B|h_{t+1})}{\Pr(A|h_{t+1})}=\frac{\Pr(B|h_t)\Pr(\alpha_t|B,h_t)}{\Pr(A|h_t)\Pr(\alpha_t|A,h_t)}=\lambda_t\frac{\Pr(\alpha_t|B,h_t)}{\Pr(\alpha_t|A,h_t)}.
\end{eqnarray}
Then we could verify that $\{\lambda_t\}$ forms a martingale:
\begin{eqnarray}
	E[\lambda_{t+1}|A,h_t]=\sum_{\alpha_t\in \{a,b\}} \bigg(\lambda_t\frac{\Pr(\alpha_t|B,h_t)}{\Pr(\alpha_t|A,h_t)}\bigg)\Pr(\alpha_t|A,h_t)=\lambda_t.
\end{eqnarray}
Therefore, martingale $\{\lambda_t\}$ converges almost surely to a limit random variable $\lambda_\infty$. This means that if we randomly draw an infinite history $h$ out, then 
\begin{eqnarray}
	\label{eqn13}
	\lim_{t\to +\infty}{\lambda_t(h)}=\lambda_\infty(h)
\end{eqnarray}
for almost all $h$. This means that the public belief updated along history $h$ eventually settle down to the value $\lambda_\infty(h)$ for almost all $h$.
Any value that r.v. $\lambda_\infty$ takes with positive probability is referred as a long run learning belief.

Therefore, after introducing a public observable payoff shock, the classic result of observational learning still holds: the public belief eventually settles down almost surely.

\section{Different Types of Long Run Learning}
In this section we classify all the long run learning beliefs into three classes: (1) confounded learning; (2) herding learning; and (3) mixed learning. 

We start with confounded learning. We recall its key property is that both types still use private information to decide at confounded learning beliefs(active). Therefore, at a confounded learning belief $\lambda^*$, the cutoffs of both types' strategy must be within the domain of private signals. That is,
\begin{eqnarray}
	\label{eqn14}
	m(\lambda^*,0),m(\lambda^*,c),mm(\lambda^*,0),mm(\lambda^*,c)\in (\underline{s},\overline{s}).
\end{eqnarray}
Condition \ref{eqn14} can be rewritten as following. Let
\begin{eqnarray}
	\mathcal{E}_0=\{\lambda|m(\lambda,0),m(\lambda,0)\in (\underline{s},\overline{s})\}\,;\, \mathcal{E}_c=\{\lambda|m(\lambda,c),m(\lambda,c)\in (\underline{s},\overline{s})\}.
\end{eqnarray}
Here $\mathcal{E}_0$ is the set of all the public beliefs such that both types are active if no shock arrives. We call $\mathcal{E}_0$ to be the confounding region without shock, and similarly call $\mathcal{E}_c$ to be the confounding region with shock $c$. Then 
\begin{definition}
	A long run learning belief $\lambda^*$ is confounded iff $\lambda^*\in \mathcal{E}_0\cap \mathcal{E}_c$.
\end{definition}
The necessary condition for a belief $\lambda^*$ to be a long run learning belief is that $\lambda^*$ must be stationary with and without shock. 
Thus we could obtain the following necessary condition for a belief $\lambda^*$ to be confounded:
\begin{condition}
	That $\lambda^*$ is confounded learning belief in a model with shock $c$ only if $\lambda^*\in \mathcal{E}_0\cap \mathcal{E}_c$ and solves
	\begin{eqnarray}
		\label{eqn17}
		\Pr(b|B,\lambda^*,0)=\Pr(b|A,\lambda^*,0);
		\Pr(b|B,\lambda^*,c)=\Pr(b|A,\lambda^*,c).
	\end{eqnarray}
\end{condition}

Now we turn to describe herding beliefs. Such a belief must make both types of players ignore their private information and act according to the public belief (inactive) in all periods. Therefore
\begin{definition}
	A long run learning belief $\lambda^*$ is herding iff $\lambda^*\in H_0\cap H_c$. 
\end{definition}
Here 
\begin{eqnarray}
    H_0&=&\{\lambda\in (0,+\infty)|m(\lambda,0),mm(\lambda,0)\notin (\underline{s},\overline{s})\};\notag\\ H_c&=&\{\lambda\in (0,+\infty)|m(\lambda,c),mm(\lambda,c)\notin (\underline{s},\overline{s})\}
\end{eqnarray}
are the herding region without (with shock $c$) respectively.

A sharp reader may have guessed the meaning of a mixed learning. 
\begin{definition}
A long run learning belief $\lambda^*$ is mixed iff $\lambda^*\in \mathcal{E}_0\cap H_c$ (first type) or $\lambda^*\in \mathcal{E}_c\cap H_0$ (second type). 
\end{definition}
These are the beliefs that are confounded without shock and are herding under shock, or vice versa. We pay special attention to the first type mixed learning. In fact, if a confounded learning belief of a no shock model is stationary in the shock model, it must either be confounded or be first type mixed. To see this, we classify all the interior public beliefs under shock $c$ as the herding region and the non-herding region $NH_c=(0,+\infty)-H_c$. For any belief in $NH_c-\mathcal{E}_c$, exactly one type of player is active. Since the actions of this type are not confounded by the other type, they are informative. Hence any belief in $NH_c-\mathcal{E}_c$ cannot be stationary.

\section{Fragility of Confounded Learning}
In this section we state our main results.
\begin{theorem}
	\label{thm6}
	For almost all private signals and almost all payoff shocks, the shock model admits no confounded learning, despite players have sufficiently divergent preferences. 
	
	Furthermore, the shock model admits no first type mixed learning, for all private signals and all small payoff shocks.
	
	To summarize, in almost all cases, none of the confounded learning beliefs in a no shock model stays stationary in the corresponding shock model.
\end{theorem}

We spend the rest of this section to set-up the necessary math so that we can rigorously translate the above theorem. 
In particular, we wan to explain the meaning of \textit{almost all private signals} and \textit{almost all payoff shocks}.

Let us recall that an observational learning model with payoff shock is described by the following parameters: (1) $p$-the proportion of match type players; (2) $u,v$-the payoff parameters of both types; (3) $(\underline{s},\overline{s})$-the strength of private signal; (4) $(f^A(s),f^B(s))$-the private signal. We first nail down those parameters that makes our main theorem meaningful. Obviously, if the confounding region without shock $\mathcal{E}_0$ is empty, then naturally there is no confounded learning or mixed learning of first type. 
So we only consider parameter tuples $(p,u,v,\underline{s},\overline{s})$ with $\mathcal{E}_0\neq \emptyset$.

After arbitrarily pick a meaningful parameter tuple, we consider all private signals with support $(\underline{s},\overline{s})$ and satisfying assumption \ref{assumption1}. We claim that this set of private signals can be identified with the following subset of $L_\infty(\underline{s},\overline{s})\cap C^0(\underline{s},\overline{s})$
\begin{eqnarray}
	\label{eqn18}
	F^\infty_{(\underline{s},\overline{s})}=\bigg\{f(s)\bigg| 
	\int_{\underline{s}}^{\overline{s}} f(s)\frac{1-2s}{s}ds=0; 
	\int_{\underline{s}}^{\overline{s}} f(s)ds=1; f(s)\geq 0;
	f(s)\frac{1-s}{s}\in L_\infty(\underline{s},\overline{s})\bigg\}.
\end{eqnarray}
To see this, we recall that private signal is modeled as direct signal-a realization of value $s$ means that ``probability of state being $A$ is $s$". Then we must have
\begin{eqnarray}
	\label{eqn19}
	\frac{f^B(s)}{f^A(s)}=\frac{1-s}{s}
\end{eqnarray}
on the common support of $(f^A(s),f^B(s))$. Equation \ref{eqn19} is necessary because it impose the requirement that the likelihood of direct signals $\frac{1-s}{s}$ must agree with the likelihood of density functions $\frac{f^B(s)}{f^A(s)}$. Obviously, a pair of density functions $(f^A(s),f^B(s))$ satisfying equation \ref{eqn19} gives out a description of direct signals. Using this relation we could identify a pair of density functions $(f^A(s),f^B(s))$ with $f^A(s)$. The $F_{(\underline{s},\overline{s})}^\infty$ is the set of density functions $f(s)$ that could serve as $f^A(s)$ in a pair. In other words, any $f(s)\in F_{(\underline{s},\overline{s})}^\infty$ is a density function satisfying that $f(s)\frac{1-s}{s}$ is also a continuous uniformly bounded density function. This signal set $F_{(\underline{s},\overline{s})}^\infty$ is embedded into $L_\infty(\underline{s},\overline{s})$ and is endowed with the subspace topology induced by $L_\infty$-norm. Now we can nail down the meaning of \textit{almost all signals}:
\begin{definition}
Holding a meaning parameter tuple $(p,u,v,\underline{s},\overline{s})$, we say a statement $P$ holds for almost all signals if it holds for all the signals $f(s)\in F_{(\underline{s},\overline{s})}^\infty$ except for a first category set. Here the first category is defined with respect to the topology induced by $L_\infty$-norm.
\end{definition}

Now we turn to nail down the meaning of ``almost all shocks". Given a parameter tuple $(p,u,v,\underline{s},\overline{s})$ with non-empty $\mathcal{E}_0$, it is trivial to say the model with shock $c$ admits no confounded learning if the confounding region under shock $c$ becomes empty. Therefore, we only consider shock in the set 
\begin{eqnarray}
	\label{eqn3}
	C=\{c\in (-\infty, +\infty)|\mathcal{E}_c\neq \emptyset\}-\{0\}.
\end{eqnarray}
We observe that this shock set $C$ must be bounded. Actually, $C$ must be a subset of $(-1,u)\cap (-v,1)$, which is the set of all shocks that guarantees both types are active. On the other hand, we also observe that $C$ must be non-empty. This is because we have guaranteed that $\mathcal{E}_0$ is non-empty by choosing the meaningful parameter tuple $(p,u,v,\underline{s},\overline{s})$. A small enough shock just slightly shift the confounding region and cannot eliminate it. So $C$ always contain an interval around $0$.
We also need a technical assumption:
\begin{assumption}
\label{assumption8}
	We arbitrarily pick a small enough $\varepsilon_0>0$ and consider 
	\begin{eqnarray}
		C_{\varepsilon_0}=C\cap (-1+\varepsilon_0,u-\varepsilon_0)\cap(-v+\varepsilon_0,1-\varepsilon_0)
	\end{eqnarray}
	to be the set of meaningful shocks.
\end{assumption}
The purpose of this assumption is to keep cutoffs $m(\lambda,c),mm(\lambda,c)$ bounded away from $0$ and $+\infty$ for all $c\in C_{\varepsilon_0}$. Given that we can take $\varepsilon_0$ arbitrarily small, we consider this as a mild assumption. 
\begin{definition}
	We say a statement holds for almost all shocks if it holds a.e. on $C_{\varepsilon_0}$.
\end{definition}

Now we can translate the first part of our main theorem into rigorous mathematical language:
\begin{theorem}
	\label{thm10}
	Given any parameter tuple $(p,u,v,\underline{s},\overline{s})$ with non-empty $\mathcal{E}_0$, for almost all private signals, the following equation system 
	\begin{eqnarray}
		\label{eqn22}
		\Pr(b|B,\lambda^*,0)=\Pr(b|B,\lambda^*,0),
		\Pr(b|B,\lambda^*,c)=\Pr(b|B,\lambda^*,c);
	\end{eqnarray}
has no solution on $\mathcal{E}_0\cap \mathcal{E}_c$ for almost all $c$.
\end{theorem}
The proof of this theorem is complicated and is delayed to the next section. 

The second part of theorem \ref{thm6} can be easily proved. If private signal is unbounded, then the herding region under shock $c$ is empty and the result is immediate. Otherwise, a small shock $c$ just slightly shift the non-herding region from $NH_0$ to $NH_c$. Because $NH_0$ strictly covers $\mathcal{E}_0$
\footnote{The non-herding region is the union of active regions of both type while the confounding region is the intersection of these active regions.}, a slight shift doesn't change it. So $\mathcal{E}_0\cap H_c$ must be empty for small $c$. In fact, given all other primitives, we call a shock $c$ \textit{small} iff it satisfies that $\mathcal{E}_0\cap H_c=\emptyset$. 

We close this section with the following corollary:
\begin{corollary}
	If private signal is of unbounded strength, in the long run all the weight is assigned to the true state for almost all signals and almost all shocks.
\end{corollary}
To see it, we observe that $\mathcal{E}_c$ is actually $(0,+\infty)$ for all $c$ in $C_{\varepsilon_0}\cup \{0\}$. Then theorem \ref{thm10} implies that the only stationary beliefs are $\lambda^*=0$ and $\lambda^*=+\infty$. Because the standard martingale argument rules out the possibility that all the weight is assigned to the incorrect state
\footnote{
Following Fatou's lemma, $\lim_{t\to +\infty} E[\lambda_t]\geq E[\lim_{t\to\infty} \lambda_t]$. If $\lambda^*$ could be $+\infty$ with positive probability, then the right-hand side is $+\infty$. This is impossible since the left-hand side is $\lambda_0$, which is finite.
}
, the result follows.

\section{Proof of Main Result}
In this section we prove our theorem \ref{thm10}. Before proof, we first rewrite the equation system \ref{eqn22} so that the notation reflect the underlying private signal $f(s)$: first, we have
\begin{eqnarray}
	&&\Pr(b|B,\lambda^*,c)-\Pr(b|A,\lambda^*,c)\notag\\
	&=&p[F^B(m(\lambda^*,c))-F^A(m(\lambda^*,c))]+(1-p)[F^B(mm(\lambda^*,c))-F^A(mm(\lambda^*,c))].\notag
\end{eqnarray}
Then use that private signal pair $(f^A(s),f^B(s))$ is given by $(f(s),f(s)\frac{1-s}{s})$ and rewrite probability as integral of density function, we have 
\begin{eqnarray}
	\label{eqn25}
	&&p[F^B(m(\lambda^*,c))-F^A(m(\lambda^*,c))]+(1-p)[F^B(mm(\lambda^*,c))-F^A(mm(\lambda^*,c))]\notag\\
	&=& p \int_{\underline{s}}^{m(\lambda^*,c)} f(s)\frac{1-2s}{s} ds+(1-p) \int_{mm(\lambda^*,c)}^{\overline{s}} f(s)\frac{1-2s}{s}ds
\end{eqnarray}
provided that $m(\lambda^*,c),mm(\lambda^*,c)\in (\underline{s},\overline{s})$. We shall denote the last line of equation \ref{eqn25} as $G_f(\lambda^*,c)$. Then the equations system \ref{eqn22} can be rewritten as 
\begin{eqnarray}
	G_f(\lambda^*,0)=G_f(\lambda^*,c)=0.
\end{eqnarray}

\subsection{Strongly Bounded Signals}
We start with the case where $0<\underline{s}<\overline{s}<1$. We name all these signals as ``strongly bounded".
For each private signal $f(s)$, let 
\begin{eqnarray}
	S_f=\{c\in C_{\varepsilon_0}|G_f(\lambda^*,0)=G_f(\lambda^*,c)=0 \mbox{ for some }\lambda^*\in \mathcal{E}_0\cap \mathcal{E}_c\}.
\end{eqnarray}
In other words, $S_f$ is the set of shocks such that confounded learning belief $\lambda^*$ arises in a shock model. We call $S_f$ as the stationary set for signal $f(s)$.

We make the following definition
\begin{definition}
	\label{def6}
	$F_n\subset F_{(\underline{s},\overline{s})}^\infty$ is the set of private signals whose stationary set $S_f$ has measure no less than $\frac{1}{n}$.
\end{definition}
Our purpose is to prove that $F_n$ is nowhere dense in $F_{(\underline{s},\overline{s})}^\infty$. Then, except a first category set $\bigcup_{n} F_n$, the stationary set of each private signal $f(s)$ must have measure zero. This completes the proof of theorem \ref{thm10} for strongly bounded signals.

To prove the nowhere denseness of $F_n$, we turn to show that $F_n$'s closure in $F_{(\underline{s},\overline{s})}^\infty$, which we denoted as $\overline{F}_n$ has no interior point. To do so, we must first characterize the closure $\overline{F}_n$. The following proposition says that $\overline{F}_n$ is almost $F_n$.

\begin{proposition}
	\label{prop7}
	If $f(s)\in \overline{F}_n$, then there exists $S\subset C_{\varepsilon_0}$ with $m(S)\geq \frac{1}{n}$ such that $\forall c\in S$, equation system
	\begin{eqnarray}				
		G_f(\lambda^*,0)=G_f(\lambda^*,c)=0
	\end{eqnarray}
has solutions on $\overline{\mathcal{E}}_0\cap \overline{\mathcal{E}}_c$.
\end{proposition}
Comparing definition \ref{def6} and proposition \ref{prop7}, we see the only difference is that: if $f(s)\in \overline{F}_n$, then the equation system could have solution on the closure of $\mathcal{E}_0\cap \mathcal{E}_c$. \\
\begin{proof}
	For any $f(s)\in \overline{F}_n$, there exists $\{f_n(s)\}\in F_n$ such that $\|f_n-f\|_{L_\infty}\to 0$. From definition \ref{def6}, each $f_n(s)$ has an associated stationary set $S_n$ with $m(S_n)\geq \frac{1}{n}$. We let 
	\begin{eqnarray}
		S_{i.o.}=\bigcap_{k}\bigcup_{k\geq n} S_n
	\end{eqnarray}
to be shocks that shows up infinite times in $\{S_n\}$. We could prove that $m(S_{i.o.})\geq \frac{1}{n}$. 

Arbitrarily choose a $c_0\in S_{i.o.}$, there exists subsequence $n_k$ such that $c_0\in S_{n_k}$. Since $c_0$ is a stationary shock for $f_{n_k}(s)$, for each $n_k$, there exists $\lambda^{c_0}_{n_k}$ being confounded belief for $f_{n_k}(s)$. That is, 
\begin{eqnarray}
	G_{f_{n_k}}(\lambda^{c_0}_{n_k},c_0)=G_{f_{n_k}}(\lambda^{c_0}_{n_k},0)=0 \mbox{ and } \lambda_{n_k}^{c_0}\in \mathcal{E}_0\cap \mathcal{E}_{c_0}
\end{eqnarray}
Then $\{\lambda_{n_k}^{c_0}\}$ has at least one cluster point $\Lambda_{c_0}$ on $\overline{\mathcal{E}}_0\cap \overline{\mathcal{E}}_{c_0}$. Let $n_{k_q}$ be a sub-subsequence of $n_k$ such that $\lambda_{n_{k_q}}^{c_0}$ converges to $\Lambda_{c_0}$. We claim that
\begin{eqnarray}
	\label{eqn29}
	\lim_{n_{k_q}\to +\infty} |G_{f_{n_{k_q}}}(\lambda^{c_0}_{n_{k_q}},c_0)-G_{f}(\Lambda_{c_0},c_0)|=\lim_{n_{k_q}\to +\infty} |G_{f_{n_{k_q}}}(\lambda^{c_0}_{n_{k_q}},0)-G_{f}(\Lambda_{c_0},0)|=0.
\end{eqnarray}
This claim directly implies that 
\begin{eqnarray}
	G_f(\Lambda_{c_0},c_0)=G_f(\Lambda_{c_0},0)=0 \mbox{ with } \Lambda_{c_0}\in \overline{\mathcal{E}}_0\cap \overline{\mathcal{E}}_{c_0},
\end{eqnarray}
which completes the proof of proposition \ref{prop7}.

To prove claim \ref{eqn29}, we do the following transformations. The first transformation is
\begin{eqnarray}
	\label{eqn31}
	&&|G_{f_{n_{k_q}}}(\lambda^{c_0}_{n_{k_q}},\mathfrak{c})-G_{f}(\Lambda_{c_0},\mathfrak{c})|\notag\\
	&=&|G_{f_{n_{k_q}}}(\lambda^{c_0}_{n_{k_q}},\mathfrak{c})-G_f(\lambda^{c_0}_{n_{k_q}},\mathfrak{c})+G_{f}(\lambda^{c_0}_{n_{k_q}},\mathfrak{c})-G_{f}(\Lambda_{c_0},\mathfrak{c})|\notag\\
	&=& |-p\int^{\overline{s}}_{m(\lambda_{n_{k_q}}^{c_0},\mathfrak{c})}
	[f_{n_{k_q}}(s)-f(s)]\frac{1-2s}{s}ds
	+(1-p)\int^{\overline{s}}_{mm(\lambda_{n_{k_q}}^{c_0},\mathfrak{c})}
	[f_{n_{k_q}}(s)-f(s)]\frac{1-2s}{s}ds\notag\\
	&&+p\int_{m(\Lambda_{c_0},\mathfrak{c})}^{m(\lambda_{n_{k_q}}^{c_0},\mathfrak{c})}
	f(s)\frac{1-2s}{s}ds
	+(1-p)\int^{mm(\Lambda_{c_0},\mathfrak{c})}_{mm(\lambda_{n_{k_q}}^{c_0},\mathfrak{c})}
	f(s)\frac{1-2s}{s}ds|.
\end{eqnarray}
Here we use symbol $\mathfrak{c}$ to represent either $0$ or $c_0$.
Besides, to obtain the first term in the last line of equation \ref{eqn31}, we use the fact that   
\begin{eqnarray*}
	p\int^{\overline{s}}_{m(\lambda_{n_{k_q}}^{c_0},c)}
	[f(s)-f_{n_{k_q}}(s)]\frac{1-2s}{s}ds=-p\int_{\underline{s}}^{m(\lambda_{n_{k_q}}^{c_0},c)}
	[f(s)-f_{n_{k_q}}(s)]\frac{1-2s}{s}ds
\end{eqnarray*}
which follows from 
$\int_{\underline{s}}^{\overline{s}}f(s)\frac{1-2s}{s}ds=\int_{\underline{s}}^{\overline{s}}f_{n}(s)\frac{1-2s}{s}ds=0$.

Then we make the following observation: 
\begin{eqnarray}
	\label{eqn35}
	m(\lambda_{n_{k_q}}^{c_0},\mathfrak{c}),mm(\lambda_{n_{k_q}}^{c_0},\mathfrak{c}),m(\Lambda_{c_0},\mathfrak{c}),mm(\Lambda_{c_0},\mathfrak{c})\in [\underline{b},\overline{b}]\subset (0,1)
\end{eqnarray}
Actually, the extreme values of $m(\lambda,c)$ and $mm(\lambda,c)$ on compact region $\overline{\mathcal{E}}_0\times \overline{C}_{\varepsilon_0}$ can be computed through the partial derivatives of $m(\lambda,c)$ and $mm(\lambda,c)$. We remark  observation \ref{eqn35} relys on two things: (1) $\lambda_{n_{k_q}}^{c_0},\Lambda_{c_0}\in \overline{\mathcal{E}}_0$ and
$0<\inf{\mathcal{E}}_0<\sup{\mathcal{E}}_0<+\infty$. Here we use assumption that $0<\underline{s}<\overline{s}<1$ to obtain that $\mathcal{E}_0$ is bounded away from $0,+\infty$; (2) $C_{\varepsilon_0}$ is bounded away from $-1,1,-v,u$ following assumption \ref{assumption8}.

Following this observation, we can apply triangle inequality to equation \ref{eqn31} and obtain that
\begin{eqnarray}
	\label{eqn33}
		&&|G_{f_{n_{k_q}}}(\lambda^{c_0}_{n_{k_q}},\mathfrak{c})-G_{f}(\Lambda_{c_0},\mathfrak{c})|\notag\\
		&\leq& \bigg[\max_{s\in [\underline{b},1]} \bigg|\frac{1-2s}{s}\bigg|\bigg]
		\bigg[		
		\|f_{n_{k_q}}(s)-f(s)\|_{L_\infty}+
		p \|f(s)\||m(\Lambda_{c_0},\mathfrak{c})- m(\lambda_{n_{k_q}}^{c_0},\mathfrak{c})|\notag\\
		&&+(1-p)\|f(s)\||mm(\Lambda_{c_0},\mathfrak{c})- mm(\lambda_{n_{k_q}}^{c_0},\mathfrak{c})|
		\bigg].
\end{eqnarray}
Since $\lambda_{n_{k_q}}^{c_0}\to \Lambda_{c_0}$, by the continuity of $m(\lambda,c),mm(\lambda,c)$, the last two terms in equation \ref{eqn33} vanishes. Recall that $\|f_{n_{k_q}}(s)-f(s)\|_{L_\infty}\to 0$, the first term vanishes as well. So we prove claim \ref{eqn29}.
\end{proof}

Now, to prove that $\overline{F}_n$ has no interior point, we turn to prove the following two propositions.
\begin{proposition}
	\label{prop8}
	For any $f(s)\in F_{(\underline{s},\overline{s})}^\infty$, there exists a sequence of real-analytic private signals $\{f_n^\omega(s)\}\subset F_{(\underline{s},\overline{s})}^\infty \cap C^\omega(\underline{s},\overline{s})$ converging to $f(s)$ in $L_\infty$-norm. In other words, the signals with real-analytic density functions are dense in $F_{(\underline{s},\overline{s})}^\infty$.
\end{proposition}
\begin{proposition}
	\label{prop9}
	Any real-analytic signal $f^\omega(s)$ is not in $\overline{F}_n$.
\end{proposition}
Following the above two propositions, any $f(s)\in \overline{F}_n$ can be approximated by a sequence of real-analytic signals which are not in $\overline{F}_n$. Therefore, no point in $\overline{F}_n$ could be interior.
Now we prove the above two propositions:\\
\begin{proof}[Proof of Proposition \ref{prop8}]
	Following the Whitney smooth approximation theorem (theorem 6.1.5 in \cite{Lee2000})) and the Whitney analytic approximation theorem (theorem 1.6.5 in \cite{Narasimhan1968}), for any $f(s)\in F^\infty_{(\underline{s},\overline{s})}$, there exists a sequence of $\overline{g}_n(s)\in C^\omega(\underline{s},\overline{s})$ such that
		\begin{eqnarray}
			\|\overline{g}_n(s)-\frac{f(s)}{s}\|_{L_\infty}<\frac{1}{n}.
		\end{eqnarray}
		Here $\frac{f(s)}{s}\in L_\infty(\underline{s},\overline{s})$ since we assume that $f(s)\frac{1-s}{s},f(s)\in L_\infty(\underline{s},\overline{s})$ for any $f(s)\in F^\infty_{(\underline{s},\overline{s})}$. The problem is that $\overline{g}_n(s)$ are not necessarily private signals. They may fail the four requirements listed in the description of $F_{(\underline{s},\overline{s})}^\infty$ as in \ref{eqn18}. We shall construct a sequence of analytic signals that uniformly converge to $f(s)$ based on the existence of $\{\overline{g}_n(s)\}_{n\in\mathbb{N}}$.
		
		First of all, we verify that
		\begin{eqnarray}
			\epsilon_n\equiv\int_{\underline{s}}^{\overline{s}} \overline{g}_n(s)(1-2s)ds\to \int_{\underline{s}}^{\overline{s}} \frac{f(s)}{s}(1-2s)ds=0.
		\end{eqnarray}
		Let $r(s)$ be a polynomial that satisfies $\int_{\underline{s}}^{\overline{s}} r(s)(1-2s)ds=1$.
		Let $\hat{g}_n(s)=\overline{g}_n(s)-\epsilon_n r(s)$, then  $\int_{\underline{s}}^{\overline{s}} \hat{g}_n(s)(1-2s)ds=0$.
		Let 
		\begin{eqnarray}
			\delta_n=\min\{\min_{s\in [\underline{s},\overline{s}]}{\hat{g}_n(s)},0\}.
		\end{eqnarray}
		We can verify that $-\frac{1}{n}-|\epsilon_n|\|r(s)\|_{L_\infty[\underline{s},\overline{s}]}\leq \delta_n\leq 0$ and hence $\delta_n\to 0$.
		Now we lift $\hat{g}_n(s)$ up to eliminate possible negative values. We construct a lifting function $h(s)$ as following:
		given $(\underline{s},\overline{s})\subset (0,1)$, 
		for any $a<0,b>1$, there exists an $\alpha\in \mathbb{R}$ such that 
		\begin{eqnarray}
			\label{eqn37}
				\int_{\underline{s}}^{\overline{s}} \frac{1}{2}(x-a)(b-x)e^{\alpha(x-\frac{1}{2})}(1-2x)dx=0.
		\end{eqnarray}
	We let $h(x)=\frac{1}{2}(x-a)(b-x)e^{\alpha(x-\frac{1}{2})}$. Intuitively, a large positive $\alpha$ makes $h(s)(1-2s)$ large on $(\frac{1}{2},\overline{s})$, where $h(s)(1-2s)$  is negative and small on $(\underline{s},\frac{1}{2})$, where $h(s)(1-2s)$ is positive; and a large negative $\alpha$ has the opposite effect on $h(s)(1-2s)$. By continuity of $\alpha$, some intermediate value makes equation \ref{eqn37} holds.
	
		Let $\delta=\min_{s\in [0,1]} h(s)$. We could verify that $\delta>0$. We let
		\begin{eqnarray}
			\tilde{g}_n(s)=\hat{g}_n(s)-\frac{\delta_n}{\delta} h(s).
		\end{eqnarray}
		We could verify that $\tilde{g}_n(s)\geq 0$. Furthermore, since equation \ref{eqn37} holds, we have $\int_{\underline{s}}^{\overline{s}} \hat{g}_n(s)(1-2s)ds=\int_{\underline{s}}^{\overline{s}} \tilde{g}_n(s)(1-2s)ds=0$.
		
		Finally, we let 
		\begin{eqnarray}
			g_n(s)=\frac{\tilde{g}_n(s)s}{\int_{\underline{s}}^{\overline{s}}\tilde{g}_n(s)sds}=\frac{[\overline{g}_n(s)-\epsilon_nr(s)-\frac{\delta_n}{\delta} h(s)]s}{\int_{\underline{s}}^{\overline{s}} [\overline{g}_n(s)-\epsilon_nr(s)-\frac{\delta_n}{\delta} h(s)]sds}.
		\end{eqnarray}
		We can verify that $\int_{\underline{s}}^{\overline{s}} g_n(s)ds=1, \int_{\underline{s}}^{\overline{s}} g_n(s)\frac{1-2s}{s}ds=0, g_n(s)\geq 0, g_n(s)\in C^\omega(\underline{s},\overline{s})$ and $g_n(s)\frac{1-s}{s}\in L_\infty(\underline{s},\overline{s})$. Therefore, $\{g_n(s)\}$ is a sequence of analytic signals. It is direct to verify that $\|g_n(s)-f(s)\|_{L_\infty}\to 0$.
	\end{proof}

\begin{proof}[Proof of Proposition \ref{prop9}]
For any real analytic signal $f(s)$, we could verify that
	\begin{eqnarray}
		G_1(x)=p\int_{\underline{s}}^x f(s)\frac{1-2s}{s}ds &;& G_2(x)= (1-p) \int_x^{\overline{s}} f(s)\frac{1-2s}{s}ds
	\end{eqnarray}
	are analytic on $(\underline{s},\overline{s})$.
	For any $\lambda\in (0,+\infty)$, let 
	\begin{eqnarray}
		M_\lambda=\{c\in C_{\varepsilon_0}|m(\lambda,c),mm(\lambda,c)\in [\underline{s},\overline{s}]\}.
	\end{eqnarray}
	we could verify that $m(\lambda,c),mm(\lambda,c)\in (\underline{s},\overline{s})$ for all $c\in \mathring{M}_\lambda$.	 Since composition of analytic functions is analytic, for any $\lambda\in (0,+\infty)$, we have
	\begin{eqnarray}
		G_f(\lambda,c)=G_1(m(\lambda,c))+G_2(mm(\lambda,c)) 
	\end{eqnarray}
is analytic in $c$ on $\mathring{M}_\lambda$, provided that $\mathring{M}_\lambda\neq \emptyset$. Similarly, for any $c\in C_{\varepsilon_0}$, $G_f(\lambda,c)$ is analytic in $\lambda$ on $\mathcal{E}_c$.
	
	It is well-known that a real-analytic function could have at most finitely many zeros on a compact interval. 
	So $G_f(\lambda,0)=0$ could have at most finitely solutions on $\mathcal{E}_0$. Given that $\overline{\mathcal{E}}_0-\mathcal{E}_0$ contains two points, we know 
	\begin{eqnarray}
		\Lambda_f=\{\lambda\in \overline{\mathcal{E}}_0|G_f(\lambda,0)=0\}
	\end{eqnarray}
is a finite set.

Recall that $f(s)\in \overline{F}_n$ implies that the  set 
\begin{eqnarray}
	\label{eqn44}
	S=\{c\in C_{\varepsilon_0}|G_f(\lambda,0)=G_f(\lambda,c)=0, \mbox{for some } \lambda\in \overline{\mathcal{E}}_0\cap \overline{\mathcal{E}}_c\}
\end{eqnarray} 
must satisfy $m(s)\geq \frac{1}{n}$.
Obviously, if $\lambda$ is a confounded belief in a shock model with shock $c\in S$, then $\lambda\in \Lambda_f$.
 Then we could divide $S$ into finitely many subsets as following: let
\begin{eqnarray}
	S_\lambda=\{c\in S|G_f(\lambda,0)=G_f(\lambda,c)=0 \mbox{ for one } \lambda\in \Lambda_f\}.
\end{eqnarray}
Then we could verify that $S=\bigcup_{\lambda\in \Lambda_f} S_\lambda$. As $m(S)\geq \frac{1}{n}$ and $\Lambda_f$ is finite, there exists at least one $\lambda_0\in \Lambda_f$ such that $m(S_{\lambda_0})>0$. 

Therefore, if an real-analytic signal $f(s)\in \overline{F}_n$, then there exists a $\lambda_0\in \overline{\mathcal{E}}_0$ and a $S_{\lambda_0}\subset C_{\varepsilon_0}$ such that
\begin{eqnarray}
	\label{eqn46}
	G_f(\lambda_0,0)=G_f(\lambda_0,c)=0, \forall c\in S_{\lambda_0};  \mbox{ and } \lambda_0\in \bigcap_{c\in S_{\lambda_0}} \overline{\mathcal{E}}_c, m(S_{\lambda_0})>0.
\end{eqnarray}
We are going to show this is impossible.

We observe that $\lambda_0\in \overline{\mathcal{E}}_c$ implies that $c\in M_{\lambda_0}$. Then $\lambda_0\in \bigcap_{c\in S_{\lambda_0}}\overline{\mathcal{E}}_c$ implies that $S_{\lambda_0}\subset M_{\lambda_0}$. We can compute $M_\lambda$ and see that it is the intersection of a (possibly degenerate) closed interval with $C_{\varepsilon_0}$. That $m(S_{\lambda_0})>0$ guarantees that $M_{\lambda_0}$ must have a non-empty interior. Then equation \ref{eqn46} implies that 
\begin{eqnarray}
	G_f(\lambda_0,c)=0; \forall c\in S_{\lambda_0}\cap \mathring{M}_{\lambda_0} \mbox{ and } m(S_{\lambda_0}\cap \mathring{M}_{\lambda_0})>0.
\end{eqnarray}
This contradict that $G_f(\lambda,c)$ is analytic in $c$ on $\mathring{M}_\lambda$.
\end{proof}

\subsection{Partially Bounded and Unbounded Signals}
In this subsection we allow $\underline{s}=0$ or (and) $\overline{s}=1$. The case that $0=\underline{s}<\overline{s}=1$ is known as private signal is of unbounded strength. If  $0=\underline{s}<\overline{s}<1$, then players could obtain arbitrarily precise signal about state being $A$ but could only obtain bounded signal about state being $B$. We name such signals as ``partially bounded". Similarly, $0<\underline{s}<\overline{s}=1$ is partially bounded as well.

We would like to extend the previous method. One difficulty arise if we try to describe $\overline{F}_n$ by running the proof of proposition \ref{prop7}. If we ask ourselves about the limit $f(s)$ of convergent sequence $\{f_n(s)\}\subset F_n$ with confounded belief $\lambda_n$ converging to $0$, we shall find proposition \ref{prop7}'s proof tells us that $f(s)$ should satisfy 
\begin{eqnarray}
	G_f(0,0)=G_f(0,c)=0, \forall c\in S\subset C_{\varepsilon_0} \mbox{ with } m(S)\geq \frac{1}{n}.
\end{eqnarray}
This condition is satisfied by all signals and has no bite.

To resolve this difficulty, we turn to prove the the signals with sufficiently small or sufficiently large confounded beliefs are small in the first category sense. In fact, we have the following two propositions.
\begin{proposition}
	\label{prop10}
	There exists $\underline{\lambda}>0$ and $\overline{\lambda}<+\infty$ such that 
	\begin{eqnarray}
		&&G_f(\lambda,0)=G_f(\lambda,c)=0 \mbox{ for some } c\in C_{\varepsilon_0} \mbox{ and some }
		\lambda\in \mathcal{E}_0\cap \mathcal{E}_c\cap \big((0,\underline{\lambda})\cup(\overline{\lambda},+\infty)\big)\notag\\
		&\Rightarrow& \exists [z_1,z_2]\subset (\underline{s},\overline{s}) s.t. f(s)=0 \mbox{ on } [z_1,z_2].
	\end{eqnarray}
\end{proposition}
In other words, if private signal $f(s)$ admits small ($<\underline{\lambda}$) or large ($>\overline{\lambda}$) confounded belief, then it must be locally zero. The next proposition states that locally zero private signals are of first category. 
\begin{proposition}
	\label{prop11}
	The set of locally zero signals 
	\begin{eqnarray}
		F_{zero}&=&\{f(s)\in F^\infty_{(\underline{s},\overline{s})}| \exists \underline{s}<z_1<z_2<\overline{s} \mbox{ s.t. } f(s)=0 \mbox{ a.e. on } [z_1,z_2].\}\notag
	\end{eqnarray}
	is of first category in $F^\infty_{(\underline{s},\overline{s})}$.
\end{proposition}

Now we prove the above two propositions.\\
\begin{proof}[Proof of Proposition \ref{prop10}]
	For any $\lambda\in \mathcal{E}_0$, 
	differentiate $G_f(\lambda,c)$ w.r.t $c$ on $\mathring{M}_\lambda$
	, we have 
	\begin{eqnarray}
		\frac{\partial}{\partial c} G_f(\lambda,c)=
		&&p f(m(\lambda,c))\frac{1-2m(\lambda,c)}{m(\lambda,c)}\frac{\partial}{\partial c} m(\lambda,c)\notag\\
		&&-(1-p)f(mm(\lambda,c))\frac{1-2mm(\lambda,c)}{mm(\lambda,c)}\frac{\partial}{\partial c}mm(\lambda,c).\notag
	\end{eqnarray}
	It is easy to verify that on $(\lambda,c)\in (0,+\infty)\times \{(-1,u)\cap (-v,1)\}$
	\begin{eqnarray}
	\label{eqn50}
		&&\frac{\partial}{\partial c} m(\lambda,c)=\frac{\lambda(u+1)}{[\lambda(1+c)+(u-c)]^2}>0; \frac{\partial}{\partial c} mm(\lambda,c)=\frac{-\lambda(v+1)}{[\lambda(1-c)+(v+c)]^2}<0.
	\end{eqnarray}
	Therefore, we have 
	\begin{eqnarray}
		\frac{\partial}{\partial c}G_f(\lambda,c)
		\begin{cases}
			\geq 0 \mbox{ on } \mathring{M}_\lambda\Leftarrow m(\lambda,c),mm(\lambda,c)<\frac{1}{2} \mbox{ on }  \mathring{M}_\lambda \notag\\
			\leq 0 \mbox{ on } \mathring{M}_\lambda\Leftarrow m(\lambda,c),mm(\lambda,c)>\frac{1}{2} \mbox{ on }  \mathring{M}_\lambda \notag\\
		\end{cases}
	\end{eqnarray}
We could make use of partial derivatives in equation \ref{eqn50} to compute 
	\begin{eqnarray}
		m(\lambda,c),mm(\lambda,c)< \frac{1}{2}
		&\Leftarrow& \lambda<\underline{\lambda}\equiv \min\Bigg\{\frac{u-\sup C_{\varepsilon_0}}{1+\sup C_{\varepsilon_0}},\frac{v+\inf C_{\varepsilon_0}}{1-\inf C_{\varepsilon_0}}\Bigg\}
	\end{eqnarray}
	and 
	\begin{eqnarray}
		m(\lambda,c),mm(\lambda,c)> \frac{1}{2}
		&\Leftarrow& \lambda>\overline{\lambda}\equiv \max\Bigg\{\frac{u-\inf C_{\varepsilon_0}}{1+\inf C_{\varepsilon_0}},\frac{v+\sup C_{\varepsilon_0}}{1-\sup C_{\varepsilon_0}}\Bigg\}.
	\end{eqnarray}
	By assuming that $C_{\varepsilon_0}$ is bounded away from $\{-1,1,-v,u\}$, we have $\underline{\lambda}>0,\, \overline{\lambda}<+\infty$.
	Combining all the computation above, we have 
	\begin{eqnarray}
		\lambda\in (0,\underline{\lambda})&\Rightarrow& \frac{\partial}{\partial c}G_f(\lambda,c)\geq 0 \mbox{ on } \mathring{M}_\lambda;\notag\\
		\lambda\in (\overline{\lambda},+\infty)&\Rightarrow& 
		\frac{\partial}{\partial c}G_f(\lambda,c)\leq 0 \mbox{ on } \mathring{M}_\lambda.
	\end{eqnarray}

Assume $f(s)$ has a small confounded belief, that is,
\begin{eqnarray}
	&&G_f(\lambda^*,0)=G_f(\lambda^*,c_0)=0 \mbox{ for some } c_0\in C_{\varepsilon_0} \mbox{ and some }
	\lambda^*\in \mathcal{E}_0\cap \mathcal{E}_{c_0}\cap (0,\underline{\lambda}).
\end{eqnarray}
We first observe that $0, c_0\in \mathring{M}_{\lambda^*}$ since $\lambda^*\in \mathcal{E}_{c_0}\cap \mathcal{E}_0$. Furthermore, because $\mathring{M}_{\lambda^*}$ is an open interval, so $[0,c_0]\subset \mathring{M}_{\lambda^*}$.
\footnote{Without loss of generality, we assume $c_0>0$.}
Because $\lambda^*<\underline{\lambda}$, $\frac{\partial}{\partial c}G_f(\lambda^*,c)\geq 0$ on $\mathring{M}_{\lambda^*}$. So
	\begin{eqnarray}
		0=G_f(\lambda^*,c_0)-G_f(\lambda^*,0)=\int_0^{c_0}\frac{\partial}{\partial c}G_f(\lambda^*,c) dc=0
	\end{eqnarray}
	implies that $\frac{\partial}{\partial c}G_f(\lambda^*,c)=0$ a.e. on $[0,c_0]$.
	Furthermore, in the expression of $\frac{\partial}{\partial c} G_f(\lambda^*,c)$,  both $p f(m(\lambda^*,c))\frac{1-2m(\lambda^*,c)}{m(\lambda^*,c)}\frac{\partial}{\partial c} m(\lambda^*,c)$ and $-(1-p)f(mm(\lambda^*,c))\frac{1-2mm(\lambda^*,c)}{mm(\lambda^*,c)}\frac{\partial}{\partial c}mm(\lambda^*,c)$ are non-negative.
	So 
	\begin{eqnarray}
		f(m(\lambda^*,c))=0, f(mm(\lambda^*,c))=0,\, \forall c\in [0,c_0].
	\end{eqnarray}
	It is direct to verify that $Range(m(\lambda^*,c)),Range(mm(\lambda^*,c))$ where $c\in [0,c_0]$ must be closed intervals within $(\underline{s},\overline{s})$. So we have obtained that a small robust belief implies local zeroness. The proof for a large robust belief is similar.
\end{proof}\\
\begin{proof}[Proof of Proposition \ref{prop11}]
	Let 
	\begin{eqnarray}
		F_{[z_1,z_2]}&=& \{f(s)\in F^\infty_{(\underline{s},\overline{s})}|f(s)=0 \mbox{ a.e. on } [z_1,z_2]\subset (\underline{s},\overline{s})\}.
	\end{eqnarray}
	Then 
	\begin{eqnarray}
		F_{zero}=\underset{\substack{\underline{s}<z_1<z_2<\overline{s}\\z_1,z_2\in \mathbb{Q}}}{\bigcup} F_{[z_1,z_2]}.
	\end{eqnarray} 
So we turn to show that $F_{[z_1,z_2]}$	is nowhere dense in $F_{(\underline{s},\overline{s})}^\infty$.
	
	To do so, we turn to show the following: for any $f(s)\in F_{[z_1,z_2]}$ and any $\varepsilon>0$, we can construct a $g(s)$ satisfying that (1) $\|g(s)-f(s)\|_{L_\infty}<\varepsilon$ and (2) $\exists \delta>0$ such that any $h(s)\in B_{\delta}(g(s))$ is not an element in $F_{[z_1,z_2]}$.
	
	Basically, we just need to continuously lift up $f(s)$ on $[z_1,z_2]$ by a number smaller than $\varepsilon$ and simultaneously adjust $f(s)$ somewhere else to keep it a private signal. Then it is obvious any $h(s)$ that's close enough to $g(s)$ under $L_\infty$-norm cannot be within $F_{[z_1,z_2]}$.
\end{proof}

Knowing that signals admitting small or large confounded beliefs must be of first category, we just need to prove that signals with moderate confounded belief is of first category as well. This can be done using the same proof as proposition \ref{prop7} with one slight modification. Instead of $F_n$, we consider $F_n^{[\underline{\lambda},\overline{\lambda}]} \subset F_{(\underline{s},\overline{s})}^\infty$ to be the set of private signals such that 
\begin{eqnarray}
	G_f(\lambda,0)=G_f(\lambda,c)=0 \mbox{ for some } \lambda\in \mathcal{E}_0\cap \mathcal{E}_c\cap [\underline{\lambda},\overline{\lambda}]
\end{eqnarray}
for each shock $c$ in $S\subset C_{\varepsilon_0}$ with measure $m(S)\geq\frac{1}{n}$. Then the following proposition describes its closure $\overline{F}^{[\underline{\lambda},\overline{\lambda}]}_n$.

\begin{proposition}
	\label{prop12}
	If $f(s)\in \overline{F}_n^{[\underline{\lambda},\overline{\lambda}]}$, then there exists $S\subset C_{\varepsilon_0}$ with $m(S)\geq \frac{1}{n}$ such that $\forall c\in S$, equation system
	\begin{eqnarray}	
		G_f(\lambda,0)=G_f(\lambda,c)=0 		
	\end{eqnarray}
	has solutions on $\overline{\mathcal{E}}_0\cap \overline{\mathcal{E}}_c\cap [\underline{\lambda},\overline{\lambda}]$.
\end{proposition}
The method in proposition \ref{prop7} can be directly used to prove proposition \ref{prop12}. By restricting the solution to interval $[\underline{\lambda},\overline{\lambda}]$, we are out of the problem described in the beginning of this subsection. 

Furthermore, proposition \ref{prop8} and \ref{prop9} directly extend to partially bounded and unbounded signals. Therefore, we have completed the proof of theorem \ref{thm10} in these cases.

\section{Conclusion}
In this article we consider observational learning model with exogenous payoff shock. We find that even a very simple shock structure eliminates the phenomenon that long run learning could be confounded. The public belief must eventually aggregate overwhelming evidence against one person's private information. 

We think observational learning with a rich payoff shock might worth more exploration. Theoretically, for example, we can ask whether information aggregation becomes faster under a random payoff shock with sufficiently rich structure. This may be helpful in resolving the slowness of sequential observational learning when the private signals have thin tails, as pointed out in \cite{HMT2018} and  \cite{RV2019}. In reality, it is often easier to alter players' payoffs rather than improve the precision of players private signals. For example, in an environment that the public mistakenly herd on  effectiveness about vaccine, to launch a educational campaign might be less effective than simply subsidize the vaccine. 
	
\bibliographystyle{ecta}
\bibliography{main}

\end{document}